\documentclass[11pt]{article}

\usepackage{amsfonts}
\usepackage{amssymb}
\usepackage{amsmath}
\usepackage{amsthm}
\usepackage{latexsym}
\usepackage{verbatim}
\usepackage{epsfig}
\usepackage{amsbsy}
\usepackage{amscd}
\usepackage{bm}

\usepackage{graphicx}

\usepackage{ifthen} 
\usepackage{xcolor}
\usepackage[colorlinks=true, linkcolor=teal, anchorcolor=teal,
            citecolor=blue,  filecolor=blue, menucolor=blue, pagecolor=teal,											urlcolor=blue, plainpages=false, pdfpagelabels]{hyperref}

\newcommand{\mymail}[1]{\href{mailto:#1}{\texttt{#1}}}

\usepackage{fancyhdr}
\usepackage{currfile}
\fancypagestyle{firststyle}{

\cfoot{{\footnotesize \texttt{\currfilename}, version: \texttt{\today} }} 
}

\newcommand{\setauthA}[1]{\def\authA{#1}}
\newcommand{\setauthB}[1]{\def\authB{#1}}

\def\printA{\begin{tabular}{l} \authA \end{tabular}}
\def\printB{\begin{tabular}{l} \authB \end{tabular}}

\newcommand{\makemytitle}[1]{\begin{center}{\textsf{\LARGE #1}}
  \end{center}
}

\usepackage[sf,bf]{titlesec}

\usepackage{algorithmic}
\usepackage{algorithm}

\usepackage[nonamebreak,square,numbers,sort]{natbib}

\usepackage[letterpaper, textheight = 676pt]{geometry}
\providecommand{\mc}[1]{\mathcal#1}
\providecommand{\mc}[1]{\mathcal#1}

\newcommand{\R}{{\mathbb R}}

\providecommand{\T}{\top} 

\providecommand{\wt}[1]{\widetilde{#1}}
\providecommand{\wh}[1]{\widehat{#1}}

\providecommand{\nnorm}[1]{ \lVert#1 \rVert}
\newcommand{\scp}[2]{\left\langle#1, #2\right\rangle}
\newcommand{\nscp}[2]{\langle#1, #2\rangle}

\newcommand{\blanco}[1]{  }

\newcommand{\deriv}[3]{%
\ifthenelse{#1 = 1}{\frac{d\,#2}{d\,#3}}{\frac{d^{{#1}} #2}{d{#3}^{{#1}}}}
}

\newcommand{\partials}[3]{%
\ifthenelse{#1 = 1}{\frac{\partial\,#2}{\partial\,#3}}{\frac{\partial^{#1}
    #2}{\partial#3^{#1}}}
} 

\def\su{\sum_{i=1}^n}

\def \coloneq{\mathrel{\mathop:}=}

\newtheorem{theo}{Theorem}

\newtheorem{lemmachen}{Theorem}

 \newtheorem{lemma}[lemmachen]{Lemma}

\def\R{\mathbb{R}}

\def\epss{\varepsilon}

\usepackage{subfigure}

\newboolean{isjournal}
\setboolean{isjournal}{true}

\newboolean{nocolors}
\setboolean{nocolors}{false}

\newboolean{usemtpro}
\setboolean{usemtpro}{false}

\ifthenelse{\boolean{nocolors}}{\hypersetup{colorlinks=false}}{}

\makeatletter
\newcommand\footnoteref[1]{\protected@xdef\@thefnmark{\ref{#1}}\@footnotemark}
\makeatother

\setauthB{{\bfseries Felicitas J. Detmer}$^*$ \\ Department of Bioengineering }

\setauthA{{\bfseries Martin Slawski}$^{\dagger}$ \\ Department of Statistics}

\begin{document}
\thispagestyle{firststyle}

\makemytitle{{\bfseries {A Note on Coding and Standardization of Categorical Variables in (Sparse) Group Lasso Regression}}}
\vskip 3.5ex
{\large\begin{center}
\printB
\printA
\vskip1.5ex
George Mason University, Fairfax, VA 22030, USA\\[.5ex]
$^{*}$\mymail{fdetmer@gmu.edu},  $^{\dagger}$\mymail{mslawsk3@gmu.edu}
\end{center}}

\vskip 3.5ex

\begin{abstract} 
  Categorical regressor variables are usually handled by introducing a set of indicator variables, and imposing
  a linear constraint to ensure identifiability in the presence of an intercept, or equivalently, using one of
  various coding schemes. As proposed in Yuan and Lin {\small{[\emph{J.~R.~Statist.~Soc.~B, \textbf{68} (2006), 49--67}]}}, the group lasso is a natural and computationally convenient approach to perform variable selection in settings with categorical covariates. As pointed out by Simon and Tibshirani {\small{[\emph{Stat.~Sin., \textbf{22} (2011), 983--1001}]}}, "standardization" by means of block-wise orthonormalization of column submatrices each corresponding to one group of variables can substantially
  boost performance. In this note, we study the aspect of standardization for the special case of categorical predictors in detail. The main result is that orthonormalization is not required; column-wise scaling of the design matrix followed by re-scaling and centering of the coefficients is shown to have exactly the same effect. Similar reductions can be achieved in the case of interactions. The extension to the so-called sparse group lasso, which additionally promotes within-group sparsity, is considered as well. The importance of proper standardization is illustrated via extensive simulations.      
\end{abstract}

\section{Introduction}\label{sec:intro}
The treatment of categorical predictor variables is covered in many widely used textbooks on linear regression. Given
$n$ observations $\{c_1, \ldots, c_n\}$ of a categorical variable with levels $\{1,\ldots,L\}$, we can define indicator
variables $x_{i\l}$ such that $x_{i\l} = 1$ if $c_i = \l$ and $0$ otherwise, $i=1,\ldots,n$, $l=1,\ldots,L$. Denoting the
corresponding regression coefficients by $\beta_1, \ldots, \beta_L$, we note that the parameters of a linear predictor-based regression model
with intercept $\beta_0$
\begin{equation}\label{eq:indicator}
\eta_i = \beta_0 + \beta_1 x_{i1} + \ldots + \beta_L x_{iL}, \;\;\, i=1,\ldots,n,
\end{equation}
where $\eta_i$ denotes the linear predictor for observation $i=1,\ldots,n$, are not identifiable as can be seen, e.g., from the corresponding matrix representation
\begin{equation*}
\eta = [\bm{1}_n \;\, X] \;\, [\beta_0; \;\, \beta], \qquad X = (x_{i\l})_{1 \leq i \leq n,\, 1 \leq \l \leq L}, \quad \eta = (\eta_1 \ldots \eta_n)^{\T}, \quad \beta=(\beta_1 \ldots \beta_L)^{\T}
\end{equation*}
with $\bm{1}_d$ representing a vector of ones for a positive integer $d$, and $[A \;\, B]$ and $[A; \;\, B]$ denoting the column-wise respectively row-wise concatenation of matrices $A$ and $B$ having an identical number of rows respectively columns. By construction, $\bm{1}_n \in \mc{X} \coloneq \text{range}(X)$, where $\text{range}(\cdot)$ returns the column space of a matrix, and thus $[\bm{1}_n \;\, X]$ has a non-trivial null space. Identifiability of parameters can be restored by imposing linear constraints on $\beta$. Common constraints include $\beta_{r} = 0$ for some $r \in \{1,\ldots,L\}$ or $\sum_{\l = 1}^L \beta_{\l} = 0$. Model \eqref{eq:indicator} can accordingly be re-parameterized as
\begin{equation}\label{eq:coding}
\eta =   [\bm{1}_n \;\, Z]   \;\, [\gamma_0; \;\, \gamma]
\end{equation}
with the columns of $Z \in \R^{n \times (L-1)}$ forming a basis of the linear space $\mc{Z} = \{z \in \R^n:\; z = X b, \; b \in \mc{C} \}$ with $\mc{C} = \{b \in \R^L: \; a^{\T} b = 0 \}$ for $a \in \R^L$ such that $\mc{X} = \text{range}(\bm{1}_n) + \mc{Z}$; here, "$+$" denotes the sum of linear spaces. Depending on the specific choice of the linear constraint represented by $a$, the matrix $Z$ in \eqref{eq:coding} can be chosen to match common coding schemes (see Figure \ref{fig:codingschemes} for an illustration), e.g.~:\begin{figure}[t]
  \begin{center}
    $\begin{array}{ccc}
       \begin{array}{c}
         1 \\
         \\
         \\
         \\
         \\[-.5ex]
         L
       \end{array}
      \begin{bmatrix}
        1 & 0 & \ldots &  0       \\
        0  & 1  & 0 & \vdots    \\
        \vdots  & \ddots & \ddots & 0\\
         0 & \ldots &  0 & 1    \\
         0 &  \ldots  &  \ldots   & 0   
      \end{bmatrix}  & \qquad  \begin{array}{c}
         1 \\
         \\
         \\
         \\
         \\[-.5ex]
         L
                               \end{array}
       \begin{bmatrix}

         1 & 0 & \ldots &  0       \\
        0  & 1  & 0 & \vdots    \\
        \vdots  & \ddots & \ddots & 0\\
         0 & \ldots &  0 & 1    \\
         -1 &  \ldots  &  \ldots   & -1   
        
      \end{bmatrix} & \qquad  \begin{array}{c}
         1 \\
         \\
         \\
         \\
         \\[-.5ex]
         L
                              \end{array}
                                 \mbox{\small{$\begin{bmatrix}
        -1 & \ldots & \ldots & -1 \\
        1  & -1 & \ldots & \vdots \\
        0  & \ddots & \ddots & \vdots \\
        \vdots &  0    & \mbox{{\small $L-2$}}        &   -1 \\
        0 & \ldots & 0 & \mbox{{\small $L-1$}}
      \end{bmatrix}$}} \\
\\[-1.5ex]
       \text{Reference coding} & \qquad \text{Effect coding} & \qquad \text{Helmert coding}
    \end{array}$
  \end{center}
\vspace*{-2ex}
  \caption{Illustration the three coding schemes discussed below. The $l$-th row of the above matrices contains the representation of the $l$-th category, $l=1,\ldots,L$.}\label{fig:codingschemes}
\end{figure}
\vskip.75ex
\noindent \emph{Reference or treatment coding}: $Z = (z_{ij})_{1 \leq i \leq n, 1 \leq \l \leq L-1}$ is such that
$z_{ic_i} = 1$ if $c_i < L$, $i=1,\ldots,n$, while all other entries are equal to zero; without loss of generality, the $L$-th level is here taken as reference category. 
\vskip1.5ex
\noindent \emph{Effect coding}: $Z$ is such that for $i=1,\ldots,n$, we have $z_{ic_i} = 1$ and $z_{im} = 0$, $1 \leq m \leq L-1$, $m \neq c_i$ if $c_i < L$, and if $c_i = L$, then $z_{im} = -1$ for all $1 \leq m \leq L-1$.  
\vskip1.5ex
\noindent \emph{Helmert coding}: $Z$ is such that for $i=1,\ldots,n$, we have $z_{im} = -1$ for $m \geq c_i$,
$z_{im} = 0$ for $m \leq c_i-2$, and $z_{i(c_i-1)} = c_i-1$ if $c_i \geq 2$. 
\vskip2.5ex
\noindent Consider now generic model fitting problems of the form 
\begin{align}
&\min_{\beta_0, \beta} \su \mc{L}(y_i, \beta_0 + (X\beta)_i) \quad \text{subject to} \;\,\beta \in \mc{C}, \label{eq:fit_constrained}\\
&\min_{\gamma_0, \gamma} \su \mc{L}(y_i, \gamma_0 + (Z\gamma)_i), \label{eq:fit_coded}
\end{align}
where $\{ y_i \}_{i = 1}^n$ are observed responses and $\mc{L}: \R \times \R \rightarrow \R_+$ is a loss function that is strictly convex in its second argument. Let $(\wh{\beta}_0, \wh{\beta})$ and $(\wh{\gamma}_0, \wh{\gamma})$ denote the minimizers of \eqref{eq:fit_constrained} and \eqref{eq:fit_coded}, respectively. In virtue of the requirement $\mc{X} = \text{range}(\bm{1}_n) + \mc{Z}$, minimization problems \eqref{eq:fit_constrained} and \eqref{eq:fit_coded} are equivalent in terms of fit, i.e.,  
\begin{equation}\label{eq:invariance_fits}
\wh{\beta}_0 \bm{1}_n + X\wh{\beta} = \wh{\gamma}_0 \bm{1}_n + Z\wh{\gamma}
\end{equation}
independent of the specific coding scheme underlying $Z$, and also independent of whether the linear constraint on $\beta$ in \eqref{eq:fit_constrained} matches $Z$ in \eqref{eq:fit_coded}. If $\mc{C}$ is in correspondence to $Z$, it additionally holds that $\wh{\beta}_0 = \wh{\gamma}_0$.
\vskip1.5ex
\noindent \emph{Example 1.} Consider $\mc{C} = \{b \in \R^L: \bm{1}_L^{\T} b = 0\}$. Both effect coding and Helmert coding yield a matrix $Z$ matching that constraint.
\vskip.75ex 
\noindent \emph{Example 2.} Let $\mc{C}$ be as in Example 1 and let $Z$ correspond to reference coding with $L$ being the reference level. While $\mc{C}$ does not match reference coding, simple algebra shows that 
\begin{equation}\label{eq:from_coding_to_reference}
\wh{\gamma}_l = \wh{\beta}_l + \sum_{m = 1}^{L-1} \wh{\beta}_m =  \wh{\beta}_l - \wh{\beta}_L, \quad l=1,\ldots,L-1,  \qquad \wh{\gamma}_0 = \wh{\beta}_0 + \wh{\beta}_L.
\end{equation}

\subsubsection*{The group lasso penalty and categorical predictors}      
In their seminal paper \cite{YuanLin2006}, Yuan and Lin suggest the group lasso penalty as one way of performing regularization and variable selection for categorical predictors. Given $p$ categorical predictors with levels $L_j$ levels, $j=1,\ldots,p$, the linear predictor is of the form 
\begin{equation*}
\eta =  [\bm{1}_n \;\, X^{(1)} \; \ldots \; X^{(p)}] \;\, [\beta_0; \;\, \beta^{(1)}; \; \ldots \; ; \beta^{(p)}]. 
\end{equation*}   
where $X^{(j)} = (x_{i\l}^{(j)})_{1 \leq i \leq n,\, 1 \leq \l \leq L_j}$ contains the indicator variables for categorical predictor 
$j=1,\ldots,p$. Extending formulations \eqref{eq:fit_constrained} and \eqref{eq:fit_coded}, group lasso-regularized regression yields the optimization problems  
\begin{align}
&\min_{\beta_0, \{ \beta^{(j)} \}} \su \mc{L}(y_i, \beta_0  +  \textstyle{\sum_{j=1}^p (X^{(j)} \beta^{(j)})_i}) +  \lambda \displaystyle\sum_{j = 1}^p \sqrt{\text{df}_j} \nnorm{\beta^{(j)}}_2 \notag \\ 
&\qquad \; \; \quad  \text{subject to} \;\,\beta_j \in \mc{C}_j, \; j=1,\ldots,p,\label{eq:fit_constrained_grouplasso}\\[2ex]
&\min_{\gamma_0, \{\gamma^{(j)} \}} \su \mc{L}(y_i, \beta_0  + \textstyle{\sum_{j=1}^p (Z^{(j)} \gamma^{(j)})_i}) + \lambda \displaystyle\sum_{j = 1}^p \sqrt{\text{df}_j} \nnorm{\gamma^{(j)}}_2, \label{eq:fit_coded_grouplasso}
\end{align}
where $\lambda \geq 0$ is the regularization parameter, $\text{df}_j = L_j-1$ denotes the "degrees of freedom" for group $j$, $\mc{C}_j$ is the linear constraint set for $\beta_j$, and $Z^{(j)} \in \R^{n \times (L_j-1)}$ is in correspondence to $(X^{(j)}, \mc{C}_j)$, $j=1,\ldots,p$. As elaborated in \cite{YuanLin2006}, the coefficients $\{ \wh{\beta}^{(j)} \}$ and $\{ \wh{\gamma}^{(j)} \}$ minimizing \eqref{eq:fit_constrained_grouplasso} and \eqref{eq:fit_coded_grouplasso}, respectively, tend to be sparse at the group level, i.e., depending on the magnitude of $\lambda$, one may have $\wh{\beta}^{(j)} \equiv 0$ or $\wh{\gamma}^{(j)} \equiv 0$ for many $j \in \{1,\ldots,p\}$. Various statistical properties of the group lasso are studied in the literature, cf., e.g., \cite{HuangZhang2010, Lounici2011, Meier2008}.       

\section{Standardization}
As opposed to a regularization-free setting, the group lasso fit is no longer invariant under changes of the constraint sets $\{ \mc{C}_j \}_{j = 1}^p$ or the coding scheme. For example, when using reference coding, the fit will depend on the choice of the reference category. Similarly, while the constraint
sets $\mc{C}_j = \{b \in \R^{L_j}:  \bm{1}_{L_j}^{\T} b = 0\}$, $j = 1,\ldots,p$, is in full correspondence to both effect coding and Helmert coding, the fits $X^{(j)} \wh{\beta}^{(j)}$ and $Z^{(j)} \wh{\gamma}^{(j)}$, $j=1,\ldots,p,$ resulting from
\eqref{eq:fit_constrained_grouplasso} and \eqref{eq:fit_coded_grouplasso}, respectively, now generally differ and are dependent on the choice of the coding scheme. This issue is encountered for other regularization schemes than the group lasso as well, cf.~\cite{Chiquet2016}, and already arises for $p = 1$; for ease of presentation, the subsequent discussion is developed with respect to that case, writing $X^{(1)} = X$ and $Z^{(1)} = Z$ etc.
\vskip2ex
\noindent
It turns out that the above lack of invariance can be addressed by replacing the group lasso penalty term $\nnorm{\beta}_2$ respectively $\nnorm{\gamma}_2$ by
$\nnorm{\Pi^{\perp} X \beta}_2$ respectively $\nnorm{\Pi^{\perp} Z \gamma}_2$, where
$\Pi^{\perp}: \R^n \rightarrow \R^n$ denotes the projection on the orthogonal complement of
$\text{range}(\bm{1}_n)$; in other words, $\Pi^{\perp}$ equals the centering operator that subtracts the column means from each column of $X$ respectively $Z$. By construction, this replacement restores invariance according to \eqref{eq:invariance_fits} as it holds in the absence of a penalty. Note that centering via action of $\Pi^{\perp}$ becomes crucial for this invariance as the constant terms $\beta_0$ respectively $\gamma_0$ do not get penalized, hence $X$ and $Z$ need to be completely disentangled from the constant term $\bm{1}_n$.

\subsubsection*{The role of orthonormalization}
From the perspective of optimization, it is more convenient to replace
$\Pi^{\perp} X$ respectively $\Pi^{\perp} Z$ by a matrix $U \in \R^{n \times (L-1)}$ whose columns form an orthonormal basis vectors of $\text{range}(\Pi^{\perp} X) = \text{range}(\Pi^{\perp} Z)$; the matrix $U$ can be obtained from a QR decomposition or an SVD of $\Pi^{\perp} Z$ respectively $\Pi^{\perp} X$. As a result, using that $\nnorm{U \alpha}_2 = \nnorm{\alpha}_2$ for all $\alpha \in \R^{L-1}$, the optimization problems
\begin{align}
&\min_{\beta_0, \beta} \su \mc{L}(y_i, \beta_0 + (\Pi^{\perp}X\beta)_i) + \lambda \nnorm{\Pi^{\perp} X \beta}_2 \quad \text{subject to} \;\,\beta \in \mc{C}, \label{eq:fit_constrained_reg}\\
&\min_{\gamma_0, \gamma} \su \mc{L}(y_i, \gamma_0 + (\Pi^{\perp} Z\gamma)_i) + \lambda \nnorm{\Pi^{\perp} Z \gamma}_2, \label{eq:fit_coded_reg} \\
&\min_{\alpha_0, \alpha} \su \mc{L}(y_i, \alpha_0 + (U \alpha)_i) + \lambda \nnorm{\alpha}_2, \label{eq:fit_ortho_reg}
\end{align}
are equivalent. Formulation \eqref{eq:fit_ortho_reg} enables the straightforward use of popular optimization algorithms such as block coordinate descent and proximal methods. Orthonormalization in regression with group lasso regularization is discussed and recommended in \cite{Simon2012}, however without a specific treatment of the case of categorical predictors. In order to gain some intuition about the role of standardization in that case, it is instructive to compare \eqref{eq:fit_constrained_reg} with constraint set
$\mc{C} = \{b \in \R^{L}:  \bm{1}_{L}^{\T} b = 0\}$ to the ``vanilla'' group lasso formulation with penalty term $\nnorm{\beta}_2$. Recalling that $X$ equals the matrix of indicator variables
with a single one per row, we have
\begin{equation}\label{eq:penalty_expanded}
  \nnorm{\Pi^{\perp} X \beta}_2 = \left( \textstyle\sum_{l = 1}^L n_l \beta_l^2 - n (\beta^{\T} \bar{x})^2 \right)^{1/2}, \qquad
  \bar{x} \coloneq \left(n_1/n, \ldots, n_L/n \right)^{\T},
\end{equation}
where the $\{ n_l \}_{l = 1}^L$ denote the frequencies of the $L$ categories. If the proportions
of categories are perfectly balanced, i.e., $n_l = n/L$, $l=1,\ldots,L$, it holds that
$\nnorm{\Pi^{\perp} X \beta}_2 = \nnorm{\beta}_2$ as long as $\beta \in \mc{C}$. On the other hand, the two penalty terms $\nnorm{\Pi^{\perp} X \beta}_2$ and $\nnorm{\beta}_2$ can differ dramatically in unbalanced settings; the vanilla group lasso penalty $\nnorm{\beta}_2$ tends to penalize overly coefficients corresponding to categories of low frequency. As an example, consider 
\begin{equation*}
\bar{x}_1 = 1/2, \; \, \bar{x}_2 = 1/4, \; \,
\bar{x}_3 = \bar{x}_4 = 1/8, \qquad \beta_1 = 1, \; \, \beta_2 = -3, \; \, \beta_3 = 4, \;\, \beta_4 = -2, 
\end{equation*}   
in which case $\beta^{\T} \bar{x} = 0$ so that 
\begin{equation*}
\nnorm{\Pi^{\perp} X \beta}_2 = \sqrt{n/2   +  (n/4) \cdot 9 + (n/8) \cdot 16 + (n/8) \cdot4}.
\end{equation*}
Observe that here, the first indicator variable ($\beta_1$) is penalized equally to the fourth one ($\beta_4$) in accordance with how both terms affect the linear predictor $\eta$ in a mean square sense: the first term yields a change of $1$ for a fraction of $1/2$ of the total number of observations, while the fourth term yields a change of $4$ for a fraction of $1/8$. This is unlike the vanilla group lasso penalty that penalizes each term independently of the number of observations in the respective category. As a result, the performance of the latter tends to be poor in the presence of categories with low prevalence but strong effect, cf.~$\S$\ref{sec:simulations}. 

\subsubsection*{Standardization by column scaling}
It turns out that due to the simple structure of the matrix $\Pi^{\perp} X$, standardization can be performed without a QR decomposition or the SVD. 
Let $S = \text{diag}(n_1^{1/2}, \ldots, n_L^{1/2})$ denote the matrix having the roots of the frequencies of the $L$ categories on its diagonal, and let $s = (n_1^{1/2}, \ldots, n_L^{1/2})^{\T}$. Then the matrix $\Pi^{\perp} X S^{-1}$ satisfies
\begin{align}
&\text{range}(\Pi^{\perp} X S^{-1}) = \text{range}(\Pi^{\perp} X) = \text{range}(U), \label{eq:equalrange}\\
&\nnorm{\Pi^{\perp} X S^{-1} v}_2^2 = \nnorm{v}_2^2 \; \; \, \forall v \; \text{such that} \, \scp{s}{v} = 0.  \label{eq:weightedsum}
\end{align}
While property \eqref{eq:equalrange} is obvious, property \eqref{eq:weightedsum} is immediate from the fact that $X S^{-1}$ has orthonormal columns and the observation that
\begin{equation}\label{eq:crucial_binary}
\bm{1}_n^{\T} X S^{-1} v = \bm{1}_L^{\T} S^2 S^{-1} v = \scp{s}{v} = 0.  
\end{equation}
Next, we state and prove a lemma akin to Lemma 1 in \cite{LimHastie2015}. 
\begin{lemma}\label{lem:Lim_Hastie}
Consider the optimization problem  
\begin{align}\label{eq:grouplasso_centered_scaled}
 \min_{\beta_0, \beta} \su \mc{L}(y_i, \beta_0 + (\Pi^{\perp}X S^{-1} \beta)_i) + \lambda \nnorm{\beta}_2
\end{align}
Then any minimizer $(\wh{\beta}_0, \wh{\beta})$ of \eqref{eq:grouplasso_centered_scaled} satisfies $\nscp{s}{\wh{\beta}} = 0$.
\end{lemma}
\begin{proof}
Let $\wh{\beta} = \wt{\beta} + \delta$ with $0 \neq \delta \in \text{null}(\Pi^{\perp}X S^{-1}) = \{c \cdot s, \; c \in \R \}$ and $\nscp{\wt{\beta}}{s} = 0$, where $\text{null}(\cdot)$ denotes null space. We have
\begin{equation*}
  \nnorm{\wh{\beta}}_2 = \sqrt{\nnorm{\wt{\beta}}_2^2 +  \nnorm{\delta}_2^2} \geq  \nnorm{\wt{\beta}}_2, \quad \text{and} \;\,
  \Pi^{\perp}X S^{-1} \wh{\beta} = \Pi^{\perp}X S^{-1} \wt{\beta}.
\end{equation*}
Since $(\wh{\beta}_0, \wh{\beta})$ minimizes \eqref{eq:grouplasso_centered_scaled}, we must have $\wh{\beta} = \wt{\beta}$ and $\delta = 0$, which proves the claim.  
\end{proof}
\noindent In combination with observations \eqref{eq:equalrange} and \eqref{eq:weightedsum}, Lemma \ref{lem:Lim_Hastie} implies that the following optimization problems are equivalent:
\begin{align}\label{eq:chain_of_equivalences}
\begin{split}
&\min_{\beta_0, \beta} \su \mc{L}(y_i, \beta_0 + (\Pi^{\perp}X S^{-1} \beta)_i) + \lambda \nnorm{\Pi^{\perp} X S^{-1}\beta}_2, \\
&\min_{\beta_0, \beta} \su \mc{L}(y_i, \beta_0 + (\Pi^{\perp}X S^{-1} \beta)_i) + \lambda \nnorm{\beta}_2 \;\; \text{subject to} \; \scp{s}{\beta} = 0, \\
  &\min_{\beta_0, \beta} \su \mc{L}(y_i, \beta_0 + (\Pi^{\perp}X S^{-1} \beta)_i) + \lambda \nnorm{\beta}_2. 
\end{split}
\end{align}
Moreover, letting $(\wh{\alpha}_0, \wh{\alpha})$ denote the minimizer of \eqref{eq:fit_ortho_reg} using orthonormalization, we have  
\begin{equation*}
\wh{\alpha}_0 + U \wh{\alpha} = \wh{\beta}_0 + \Pi^{\perp}X S^{-1} \wh{\beta}. 
\end{equation*}
The only shortcoming of the solution $(\wh{\beta}_0, \wh{\beta})$ is limited interpretability since the constraint $\nscp{s}{\beta} = 0$ is
not as convenient as $\nscp{\bm{1}_L}{\beta} = 0$. However, there is a straightforward fix to this shortcoming: consider
\begin{equation}\label{eq:from_beta_to_theta}
\wh{\theta} = S^{-1} \wh{\beta} - \frac{\bm{1}_L^{\T} S^{-1} \wh{\beta}}{L} \bm{1}_L.
\end{equation}
Then $\bm{1}_L^{\T} \wh{\theta} = 0$ and 
\begin{equation*}
\Pi^{\perp} X \wh{\theta} =  \Pi^{\perp} X S^{-1} \wh{\beta},
\end{equation*}
hence we can work with the interpretable set of coefficients $\wh{\theta}$ without changing the fit. Moreover, $\wh{\theta}$ can be
transformed further to match a reference coding scheme, cf.~\eqref{eq:from_coding_to_reference}. Our suggested scheme is summarized
in Figure \ref{fig:standardizationscheme} as a computationally simpler alternative to standardization by orthonormalization. 

\begin{figure}
  \fbox{
\begin{minipage}{\textwidth}
    \begin{center}

\begin{enumerate}    
\item Compute the column sums $s$ of $X$
\item Compute the group lasso solution $(\wh{\beta}_0, \wh{\beta})$ with design matrix $\Pi^{\perp} X S^{-1}$\footnote{Note that the matrix
  $\Pi^{\perp} X$ does not have to be materialized.}
\item Modify $\wh{\beta}$ according to \eqref{eq:from_beta_to_theta}   
\end{enumerate}
\end{center}
\end{minipage}
}
\caption{Suggested standardization scheme for the group lasso with categorical predictors.}\label{fig:standardizationscheme}
\end{figure}

\noindent We finally remark that the reasoning of this paragraph also applies to the case of multiple categorical variables (cf.~formulations \eqref{eq:fit_constrained_grouplasso} and \eqref{eq:fit_coded_grouplasso}) at the level of each individual variable.

\subsubsection*{Handling Interaction terms}
It turns that the approach of the previous section naturally generalizes to interactions. To keep matters simple, we limit our presentation to first-order interactions; the same concepts can be used to obtain extensions to higher-order interactions. Consider two categorical variables with $L$ respectively $M$ levels, and let $X^{(1)} \in \R^{n \times L}$ and $X^{(2)} \in \R^{n \times M}$ be the corresponding indicator matrices for a sample of size $n$. Let further $X^{(1.2)} \in \R^{n \times (L \cdot M)}$ be the matrix of interactions terms obtained by entry-wise multiplication of each column of $X^{(1)}$ with each column of $X^{(2)}$. The linear predictor is given by
\begin{equation}\label{eq:interactionmodel}
\eta = [\bm{1}_n \;\,\, X^{(1)} \;\,\, X^{(2)} \;\,\, X^{(1.2)}] \;\, [\beta_0; \;\,\, \beta^{(1)}; \;\,\, \beta^{(2)}; \;\,\,  \beta^{(1.2)}],
\end{equation}
where
\begin{equation*}
\beta^{(1.2)} = \left(\beta_{11} \ldots \beta_{1M} \; \beta_{21} \ldots \beta_{2M} \; \ldots \, \ldots \beta_{L1} \ldots \beta_{LM} \right)^{\T}. 
\end{equation*}
For what follows, we shall assume that $X^{(1.2)}$ has rank $L \cdot M$. In order to ensure identifiability of the coefficients in \eqref{eq:interactionmodel}, we impose the constraints
\begin{align}\label{eq:interactionmodel_constraints}
\begin{split}  
  \textstyle\sum_{l = 1}^L \beta_l^{(1)} = 0, \quad \textstyle\sum_{m = 1}^M \beta_m^{(2)} =0, \quad &\textstyle\sum_{m = 1}^M \beta_{lm} = 0, \; 1 \leq l \leq L,\\
  &\textstyle\sum_{l = 1}^L \beta_{lm} = 0, \; 1 \leq m \leq M.
\end{split}
\end{align}
The associated group lasso problem is given by
\begin{align}\label{eq:grouplasso_interaction_plain}
\begin{split}  
  \min_{\beta_0, \, \beta^{(1)},\,  \beta^{(2)},  \,  \beta^{(1.2)}} \; \Bigg\{& \su \mc{L}(y_i, \beta_0 + (X^{(1)} \beta^{(1)})_i + (X^{(2)} \beta^{(2)})_i  + (X^{(1.2)} \beta^{(1.2)})_i)) \\
  &+ \lambda \{
  \textstyle \sqrt{\text{df}_1} \nnorm{\beta^{(1)}}_2 + \sqrt{\text{df}_2} \nnorm{\beta^{(2)}}_2 +  \sqrt{\text{df}_{1.2}} \nnorm{\beta^{(1.2)}}_2 \Bigg\} \\
  &\text{subject to} \; \eqref{eq:interactionmodel_constraints}
\end{split}
\end{align}
with degrees of freedom $\text{df}_1 = L-1$, $\text{df}_2 = M-1$, $\text{df}_{1.2} = L\cdot M - L - M + 1$. Let $\Pi^{\perp}$ be as above, and let further
$\texttt{P}^{\perp}$ denote the projection on the orthogonal complement of $\text{range}(\bm{1}_n) + \text{range}(X^{(1)}) + \text{range}(X^{(2)})$. Moreover, let
\begin{alignat*}{2}
  &S^{(1.2)} &&= \text{diag}(n_{11}^{1/2}, \ldots, n_{1M}^{1/2}, n_{21}^{1/2}, \ldots, n_{2M}^{1/2}, \ldots, \ldots, n_{LM}^{1/2}), \\[1ex]
  &S^{(1)} &&=  \text{diag}(n_{1\bullet}^{1/2},\ldots, n_{L\bullet}^{1/2}), \quad \; \qquad  S^{(2)} =   \text{diag}(n_{\bullet 1}^{1/2},\ldots, n_{\bullet M}^{1/2}),\\
  &n_{l\bullet} &&\coloneq \textstyle\sum_{m = 1}^M n_{lm}, \; 1 \leq l \leq L,   \qquad n_{\bullet m} \coloneq \textstyle\sum_{l = 1}^L n_{lm}, \; 1 \leq m \leq M, 
\end{alignat*}
where $n_{lm}$ denotes the frequency of observations with the first categorical variable being equal to $l$ and the second categorical variable being equal to
$m$, $l=1,\ldots,L$, $m=1,\ldots,M$. 
We consider the following standardized formulation of \eqref{eq:grouplasso_interaction_plain}:
\begin{align}\label{eq:grouplasso_interaction_standardized}
\begin{split}  
  \min_{\beta_0, \, \beta^{(1)},\,  \beta^{(2)},  \,  \beta^{(1.2)}} \; \Bigg\{& \su \mc{L} \Big(y_i, \beta_0 + (\Pi^{\perp} X^{(1)} [S^{(1)}]^{-1} \beta^{(1)})_i
  + (\Pi^{\perp} X^{(2)} [S^{(2)}]^{-1} \beta^{(2)})_i\\
  &\;\;\qquad + (\texttt{P}^{\perp} X^{(1.2)} [S^{(1.2)}]^{-1} \beta^{(1.2)})_i) \Big) \\
  &\;\;\qquad + \lambda \{
  \textstyle \sqrt{\text{df}_1} \nnorm{\beta^{(1)}}_2 + \sqrt{\text{df}_2} \nnorm{\beta^{(2)}}_2 +  \sqrt{\text{df}_{1.2}} \nnorm{\beta^{(1.2)}}_2 \Bigg\}
\end{split}
\end{align}
The properties of this formulation are summarized in the following result.
\begin{theo}\label{theo:interaction}
Consider optimization problem \eqref{eq:grouplasso_interaction_plain}. Then any minimizer {\small $(\wh{\beta}_0, \, \wh{\beta}^{(1)},\,  \wh{\beta}^{(2)},  \,  \wh{\beta}^{(1.2)})$}
satisfies
\begin{align}
  &\nnorm{\Pi^{\perp} X^{(1)} [S^{(1)}]^{-1} \wh{\beta}^{(1)}}_2 = \nnorm{\wh{\beta}^{(1)}}_2, \qquad \nnorm{\Pi^{\perp} X^{(2)} [S^{(2)}]^{-1} \wh{\beta}^{(2)}}_2 = \nnorm{\wh{\beta}^{(2)}}_2,  \label{eq:standardization_condition_marginal} \\[1ex] 
  &\nnorm{\emph{\texttt{P}}^{\perp} X^{(1.2)} [S^{(1.2)}]^{-1} \wh{\beta}^{(1.2)}}_2 = \nnorm{\wh{\beta}^{(1.2)}}_2 \label{eq:standardization_condition_interaction}.
\end{align}
Moreover, denote by $\{ e_j \}_{j = 1}^{M \cdot L}$ the canonical basis vectors of $\R^{M \cdot L}$, and let
\begin{align*}
&\mc{N}^{(1)} = \big\{\textstyle\sum_{m = 1}^M e_m, \textstyle\sum_{m = M+1}^{2M} e_m, \ldots, \textstyle\sum_{m = (L-1) \cdot M +1}^{L \cdot M} e_m \big\}, \\[1ex]
&\mc{N}^{(2)} = \big\{\textstyle\sum_{l = 1}^L e_{1 + (l-1) \cdot M}, \textstyle\sum_{l = 1}^{L} e_{2 + (l-1) \cdot M}, \ldots,  \textstyle\sum_{l = 1}^{L} e_{M + (l-1) \cdot M} \big\},
\end{align*}  
\end{theo}
\noindent and $\mc{N} = \mc{N}^{(1)} \cup \mc{N}^{(2)}$. Consider $\wh{\theta}^{(1.2)} = \texttt{P}_{\mc{N}}^{\perp} [S^{(1.2)}]^{-1} \wh{\beta}^{(1.2)}$, where
$\texttt{P}_{\mc{N}}^{\perp}$ denotes the orthogonal projection on $\text{range}(\mc{N})^{\perp}$. Then
\begin{align}
 &\texttt{P}^{\perp} X^{(1.2)} [S^{(1.2)}]^{-1} \wh{\beta}^{(1.2)} = \texttt{P}^{\perp} X^{(1.2)} \wh{\theta}^{(1.2)}, \label{eq:nochange_in_fit} \\
 &\textstyle\sum_{m = 1}^M \wh{\theta}_{lm} = 0, \; 1 \leq l \leq L, \quad \textstyle\sum_{l = 1}^L \wh{\theta}_{lm} = 0, \; 1 \leq m \leq M.  \label{eq:sumtooneconstraints_satisfied}  
\end{align}
\vskip1ex
\noindent A proof of Theorem \ref{theo:interaction} is provided in the appendix.
\vskip2ex
\noindent The implications of Theorem \ref{theo:interaction} for group lasso standardization for categorical variables and first-order interactions are summarized in
Figure \ref{fig:standardizationscheme_interaction}; the last step restores the usual ANOVA-type interpretation of the coefficients. The main benefit of the proposed scheme is that it does not require orthonormalization of the $L \cdot M$ matrix
$\texttt{P}^{\perp} X^{(1.2)}$. Instead, computation of the two projectors $\texttt{P}^{\perp}$ and $\texttt{P}_{\mc{N}}^{\perp}$ scales with $M + L$ rather than $M \cdot L$.

\begin{figure}
  \fbox{
\begin{minipage}{\textwidth}
    \begin{center}

\begin{enumerate}    
\item Cross-tabulate the two categorical variables to obtain frequency counts $n_{11}, \ldots, n_{LM}$.
\item Obtain $\texttt{P}^{\perp}$ from an SVD of $[X^{(1)} \; X^{(2)}]$, and obtain $\texttt{P}_{\mc{N}}^{\perp}$ by orthonormalizing $\mc{N}$. 
\item Compute the group lasso solution \eqref{eq:grouplasso_interaction_plain}.
\item Compute $\wh{\theta}^{(1)} = [S^{(1)}]^{-1} \wh{\beta}^{(1)} - \frac{\bm{1}_L^{\T} [S^{(1)}]^{-1} \wh{\beta}^{(1)}}{L} \bm{1}_L$,
              $\wh{\theta}^{(2)} = [S^{(2)}]^{-1} \wh{\beta}^{(2)} - \frac{\bm{1}_M^{\T} [S^{(2)}]^{-1} \wh{\beta}^{(2)}}{M} \bm{1}_M$  and $\wh{\theta}^{(1.2)} = \texttt{P}_{\mc{N}}^{\perp} [S^{(1.2)}]^{-1} \wh{\beta}^{(1.2)}$.   
\end{enumerate}
\end{center}
\end{minipage}
}
\caption{Suggested standardization scheme for the group lasso with categorical predictors and first-order interaction.}\label{fig:standardizationscheme_interaction}
\end{figure}

\subsubsection*{Standardization for the Sparse Group Lasso} 
Simon et al.~\cite{Simon2012b} consider an extension of the group lasso called sparse group lasso which allows for sparsity both at the group level and within each group. For a single categorical variable with sum-to-zero constraint, the optimization problem reads 
\begin{equation*}
\min_{\beta_0, \beta} \su \mc{L}(y_i, \beta_0 + (X \beta)_i) + \lambda (\tau \nnorm{\beta}_2  + (1 - \tau) \nnorm{\beta}_1 ) \quad \text{subject to} \; \bm{1}_{L}^{\T} \beta = 0. 
\end{equation*}
for a second tuning parameter $\tau \in [0,1]$. Regarding standardization, we note that the presence of the $\ell_1$ penalty breaks the chain of equivalences in \eqref{eq:chain_of_equivalences}. Specifically, for the sparse group lasso analog to \eqref{eq:grouplasso_centered_scaled}     
\begin{align}\label{eq:sparsegrouplasso_centered_scaled}
 \min_{\beta_0, \beta} \su \mc{L}(y_i, \beta_0 + (\Pi^{\perp}X S^{-1} \beta)_i) + \lambda (\tau \nnorm{\beta}_2 + (1- \tau) \nnorm{\beta}_1 )
\end{align}
there is no counterpart to Lemma \ref{lem:Lim_Hastie}, i.e., the minimizer $\wh{\beta}$ of \eqref{eq:sparsegrouplasso_centered_scaled} no longer satisfies 
$\nscp{s}{\wh{\beta}} = 0$ so that it does not hold that 
$\nnorm{\Pi^{\perp}X S^{-1} \wh{\beta}}_2 = \nnorm{\wh{\beta}}_2$. In our simulations, however, the solution of \eqref{eq:sparsegrouplasso_centered_scaled} and the solution of   
\begin{align}\label{eq:sparsegrouplasso_centered_scaled_proper}
 \min_{\beta_0, \beta} \su \mc{L}(y_i, \beta_0 + (\Pi^{\perp}X S^{-1} \beta)_i) + \lambda (\tau \nnorm{\Pi^{\perp} X S^{-1} \beta}_2 + (1- \tau) \nnorm{\beta}_1 ),
\end{align}
exhibited similar performance with regard to estimation and prediction. While \eqref{eq:sparsegrouplasso_centered_scaled_proper} is recommended from the perspective of proper standardization, formulation \eqref{eq:sparsegrouplasso_centered_scaled} appears to be more convenient for optimization.    

\section{Numerical Results}\label{sec:simulations}
We here present the results of simulations on least squares regression, i.e., $\mc{L}(z, z') = (z - z')^2/n$ that underline the importance of standardization when using the group lasso penalty for
categorical predictors in an unbalanced setting, i.e., when the frequencies of the levels exhibit strong variation. We consider
two sets of simulations, one concerning sparsity at the group level ("ordinary" group lasso), and one concerning sparsity at
both the group level and within each group (sparse group lasso). For the former, standardization as outlined in the previous
section yields a noticeable reduction of the estimation error; that effect vanishes as the level-wise frequencies become balanced.
For the sparse group lasso, we compare three different ways of standardization with varying degrees of computational ease from an optimization perspective. We
find that simple column scaling as advocated above achieves similar performance as a more sophisticated method of standardization.






\vspace*{-1ex}
\subsubsection*{Setting 1: ordinary group lasso}
\vspace*{-.5ex}
Data is generated as
\begin{equation}\label{eq:setup_grouplasso}
y = \sum_{j = 1}^{10} X^{(j)} \beta_*^{(j)} + 0.2 \cdot \epss, \qquad \epss \sim N(0, I),
\end{equation}
where $\beta_*^{(j)} \sim u$ if $1 \leq j \leq s$, and $\beta_*^{(j)} \equiv 0$ otherwise, where
$u$ is drawn from a $N(0, I_L)$-distribution ($L = 10$), then centered and scaled to unit $\ell_2$ norm;
here, $s \in \{1,2,5\}$ denotes the sparsity level. The $\{ X^{(j)} \}_{j = 1}^{10}$ equal the indicator
matrices associated with each of the $p = 10$ categorical variables. We generate a training set of
size $1,000$ and a validation set of size $200$ according to the above model. For each categorical variable,
the frequency of level $l$ is proportional to $t^l$, $l = 1,\ldots,L$, where $t \in \{2/3, 0.95\}$. Small
$t$ yields an unbalanced setting with strong heterogeneity across level-wise frequencies, whereas $t = 1$ corresponds
to the perfectly balanced case. The table of the absolute frequencies of the ten levels in the training set are provided in Table \ref{tab:frequencies}. Assignment of individual observations to the $L$ levels is done uniformly at random, separately for each categorical variable.  
\begin{table}
\begin{center}
  \begin{tabular}{|c|cccccccccc|}
    \hline
  & $n_1$ & $n_2$  & $n_3$ & $n_4$ & $n_5$ & $n_6$ & $n_7$ & $n_8$ & $n_9$ & $n_{10}$  \\
  \hline
  $t = 2/3$   & 339   & 226   & 151 &  100  & 67  & 45  & 30  & 20  & 13  & 9 \\
  \hline
  $t = 0.95$   & 125   & 118   & 112 &  107  & 101  & 96  & 92  & 87 & 83  & 79 \\
  \hline
\end{tabular}
\end{center}
\vspace*{-1.25ex}
\caption{Absolute frequencies $\{ n_l \}_{l = 1}^{10}$ in a highly unbalanced ($t = 2/3$) and nearly balanced ($t = 0.95$) setting.}\label{tab:frequencies}
\end{table}
We compare the group lasso with (i) only centering of the $\{ X^{(j)} \}_{j = 1}^p$ and with (ii) proper standardization (cf.~Figure \ref{fig:standardizationscheme} and \eqref{eq:grouplasso_centered_scaled}) with regularization parameter chosen as to minimize the mean squared prediction error on the validation set via a grid search over $\Lambda_{\textsf{w/o}} =
\sigma \cdot 2^{\kappa} \cdot (\sqrt{L/n} + \sqrt{\log(p) / n})$, $\kappa \in \{-4,-3.5,\ldots,5\}$ and $\Lambda_{\textsf{w/}} = n^{-1/2} \Lambda_{\textsf{w/o}}$ without (\textsf{w/o})
and with (\textsf{w/}) standardization, respectively. The choice of the grids follows theoretical results on the group lasso in \cite{Negahban2009}, p.~553; the
re-scaling by the factor $n^{-1/2}$ in $\Lambda_{\textsf{w/}}$ results from the fact that the results in \cite{Negahban2009} presume that each column in the design matrix
has norm of the order of $\sqrt{n}$, whereas after scaling the column norms are equal to one. Figure \ref{fig:grouplasso_results} shows boxplots of the estimation errors
$\nnorm{\wh{\beta} - \beta_*}_2$, $\beta_* = [\beta_*^{(1)}; \ldots ; \beta_*^{(p)}]$ based on 50 independent replications; in each replication, all components of \eqref{eq:setup_grouplasso} are re-generated. 
\begin{figure}[h!]
\mbox{\includegraphics[width = 0.32\textwidth]{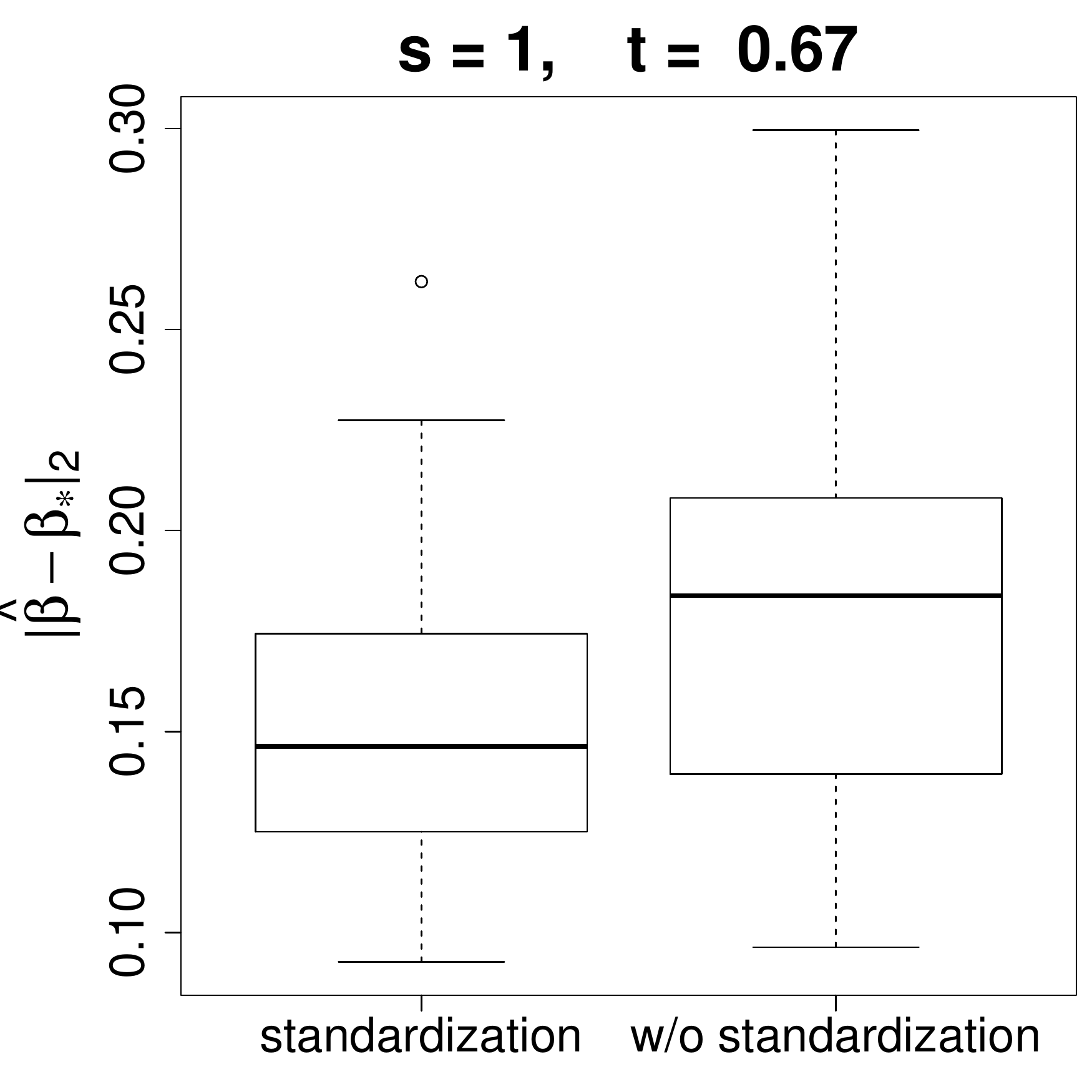} \includegraphics[width = 0.32\textwidth]{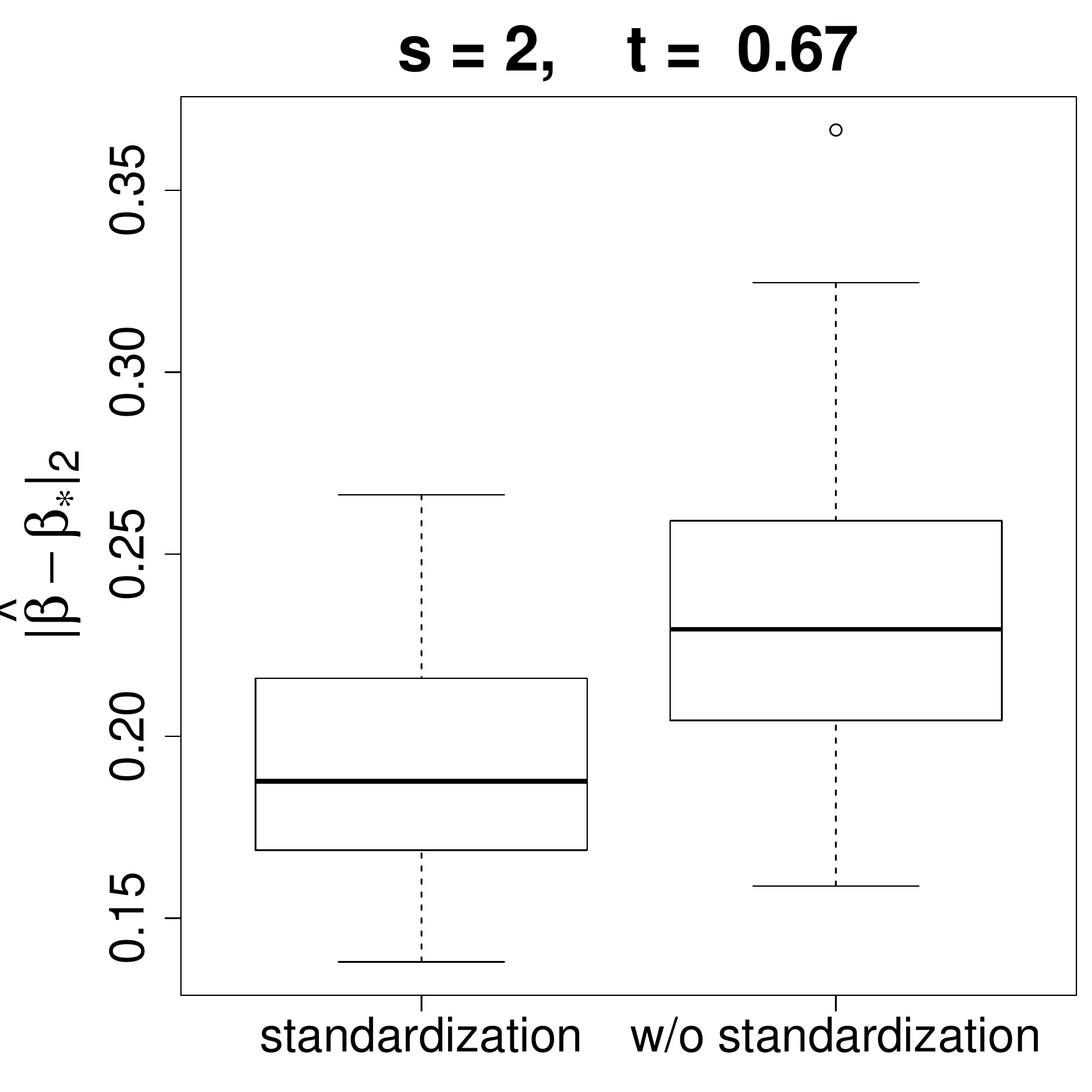} 

\includegraphics[width = 0.32\textwidth]{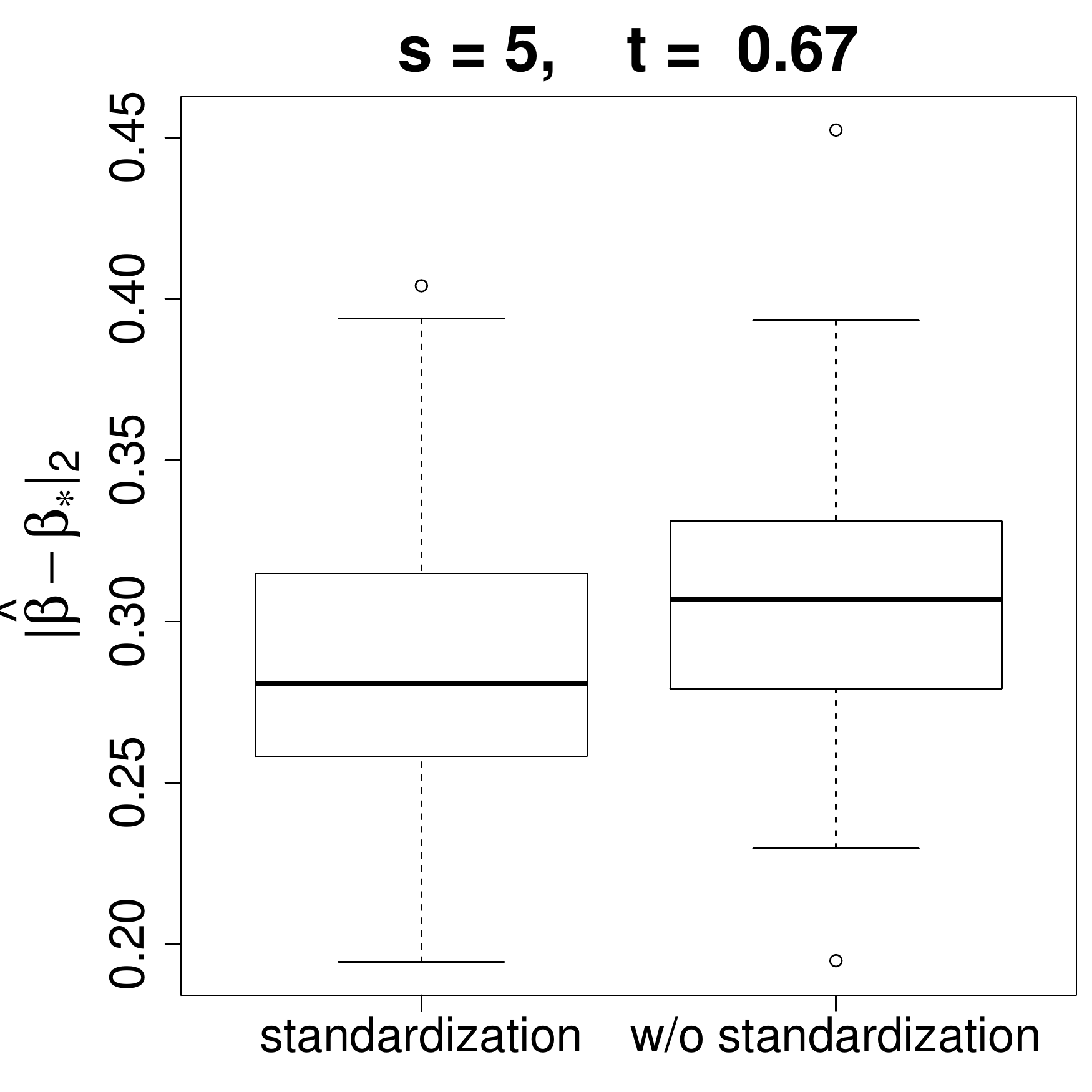}}

\mbox{\includegraphics[width = 0.32\textwidth]{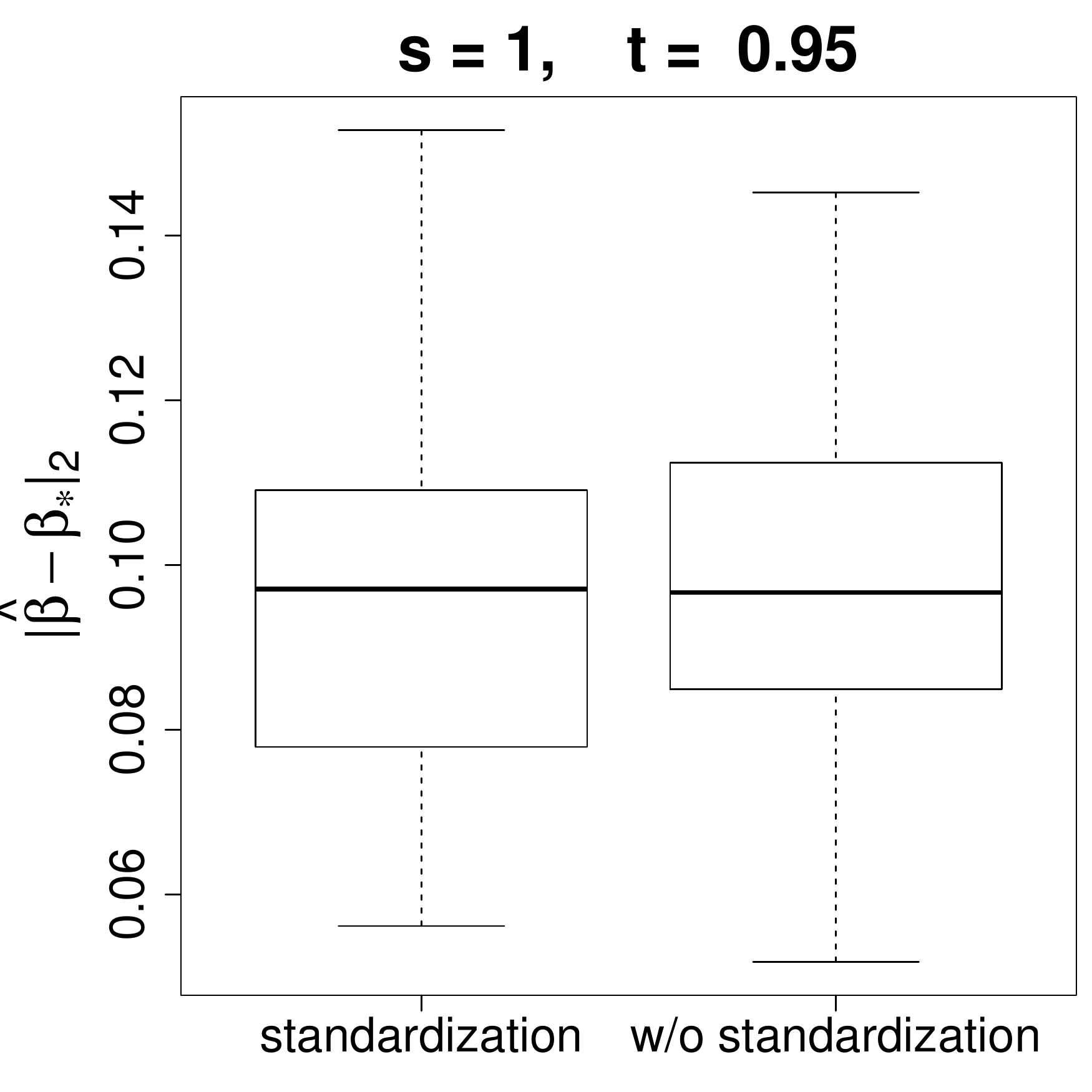} \includegraphics[width = 0.32\textwidth]{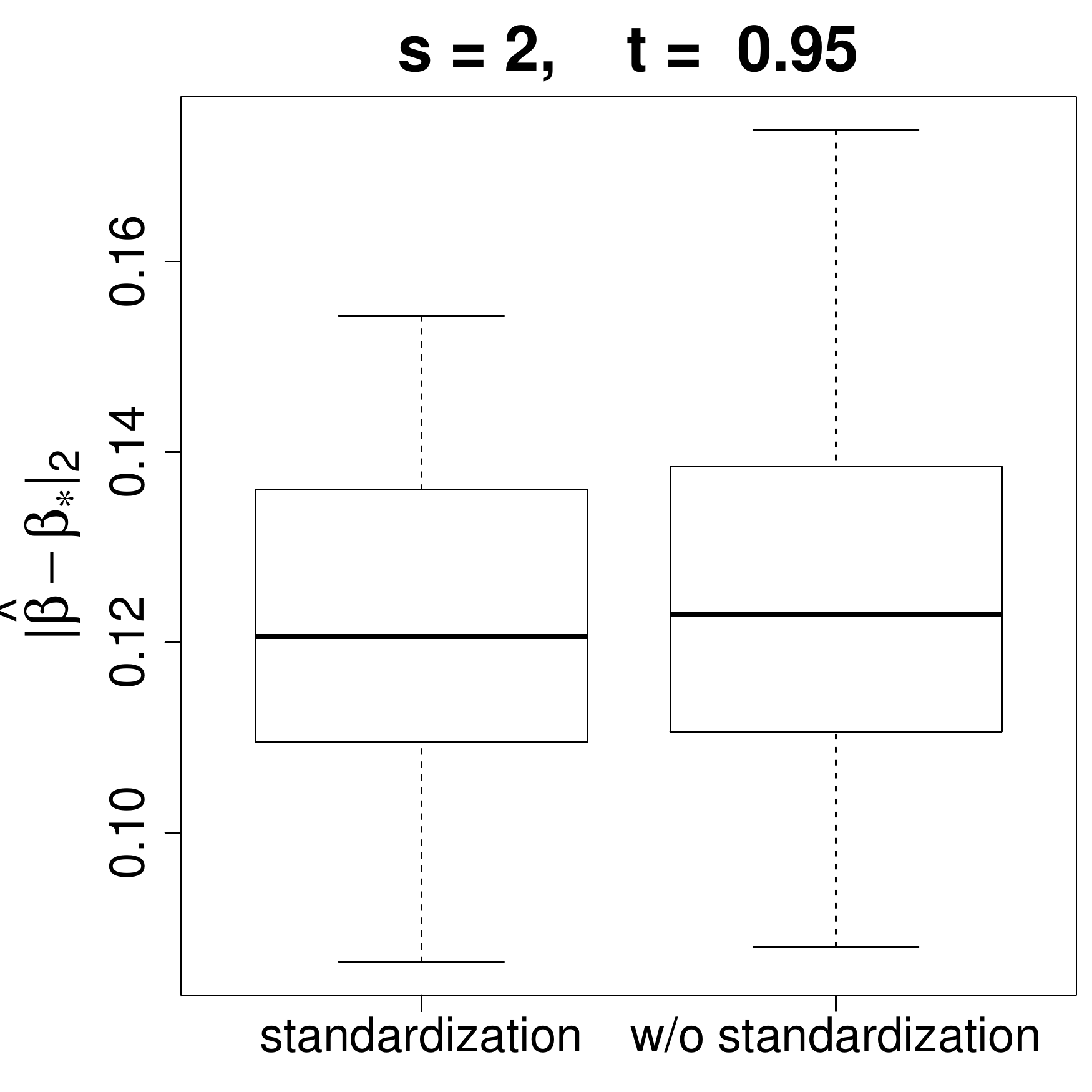} 

\includegraphics[width = 0.32\textwidth]{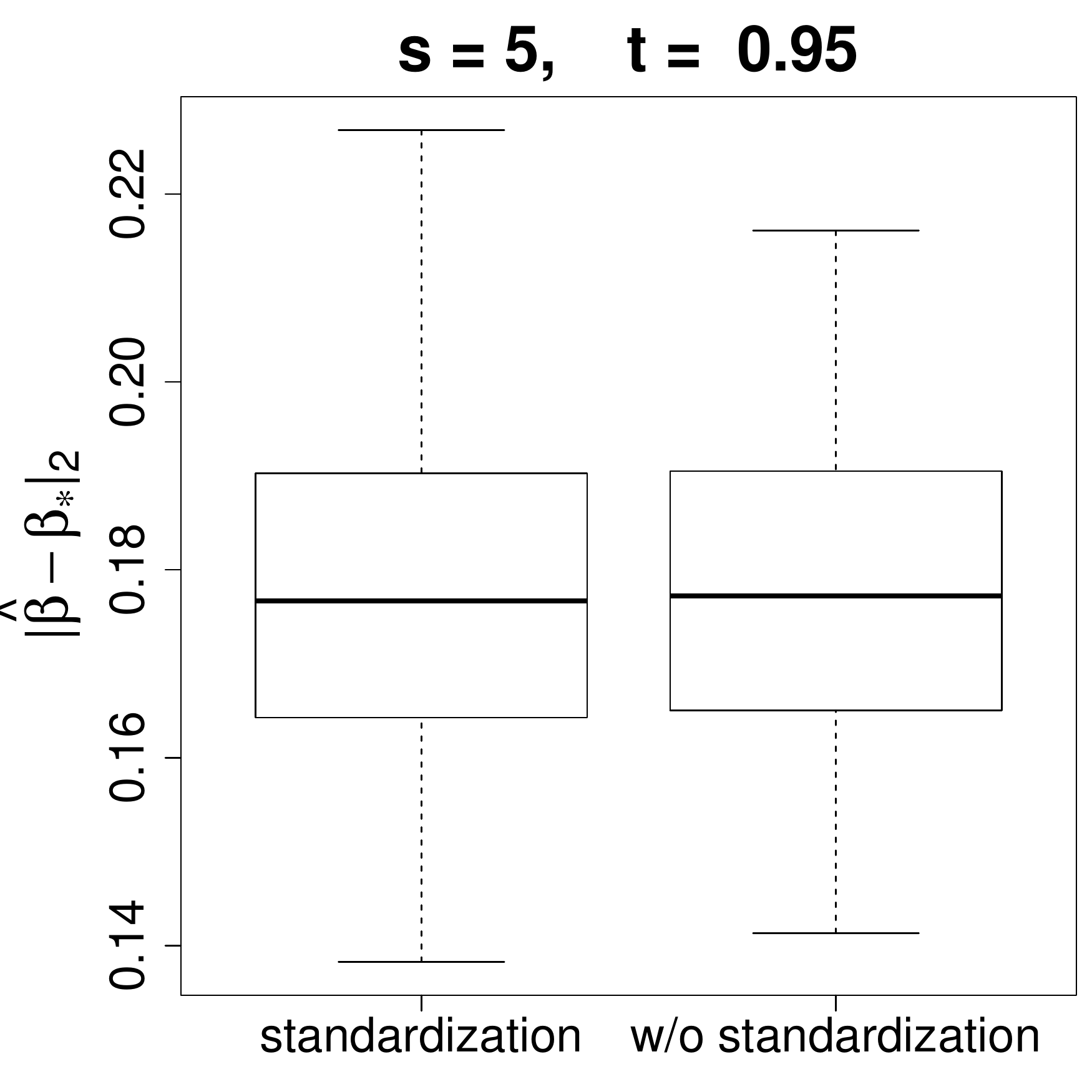}} 
\vspace*{-1.25ex}
\caption{Boxplots of the estimation errors $\nnorm{\hat{\beta} - \beta_*}_2$ for Setting 1 over $50$ independent replications. The headers
indicate the sparsity level $s$ and the frequency distribution}\label{fig:grouplasso_results}
\end{figure} 



\subsubsection*{Setting 2: sparse group lasso}
Data generation remains unchanged except for the generation of the $\beta_*^{(j)}$, $1 \leq j \leq s$, which are sparsified
simply by selecting their supports as three elements of $\{1,\ldots,L=10\}$ uniformly at random. We then compare the three
following approaches of fitting a sparse group lasso model based on the resulting data.
{\small
\begin{align}
&\text{\underline{scaling}}: \notag\\  
  &\min_{\beta_0, \{ \beta^{(j)} \}_{j = 1}^{10}} \su \mc{L}(y_i, \beta_0  +  \textstyle{\sum_{j=1}^{10}  (\Pi^{\perp} X^{(j)} [S^{(j)}]^{-1} \beta^{(j)})_i}) + \displaystyle\sum_{j = 1}^{10} \{\tau \lambda \nnorm{\beta^{(j)}}_2 +
    (1 - \tau) \lambda^{\textsf{lasso}} \nnorm{\beta^{(j)}}_1 \} \label{eq:scaling_sparsegrouplasso}
\end{align}}
{\small
\begin{align}
&\text{\underline{SVD}}: \notag \\  
  &\min_{\beta_0, \{ \beta^{(j)} \}_{j = 1}^{10}} \su \mc{L}(y_i, \beta_0  +  \textstyle{\sum_{j=1}^{10}  (\Pi^{\perp} X^{(j)} \beta^{(j)})_i}) + \displaystyle\sum_{j = 1}^{10} \{\tau \lambda \nnorm{\Pi^{\perp} X^{(j)} \beta^{(j)}}_2 +
    (1 - \tau) \lambda^{\textsf{lasso}} \nnorm{\beta^{(j)}}_1 \} \label{eq:SVD_sparsegrouplasso}
\end{align}}
\noindent Here, the designation "SVD" refers to the fact that the above optimization is equivalent to replacing $\Pi^{\perp} X^{(j)}$ by the associated matrix of left singular vectors corresponding
to the non-zero singular values, and re-expressing $\beta^{(j)}$ accordingly.
{\small
\begin{align}
&\text{\underline{SVD + scaling}}: \notag \\  
  &\min_{\beta_0, \{ \beta^{(j)} \}_{j = 1}^{10}} \su \mc{L}(y_i, \beta_0  +  \textstyle{\sum_{j=1}^{10}  (\Pi^{\perp} X^{(j)} [S^{(j)}]^{-1} \beta^{(j)})_i}) + \qquad \qquad \qquad \qquad \qquad \qquad \qquad \qquad \;\;\,\quad \notag \\
&\qquad \qquad \qquad \qquad + \displaystyle\sum_{j = 1}^{10} \{\tau \lambda \nnorm{\Pi^{\perp} X^{(j)} [S^{(j)}]^{-1} \beta^{(j)}}_2 +
(1 - \tau) \lambda^{\textsf{lasso}} \nnorm{\beta^{(j)}}_1 \} \label{eq:SVDs_sparsegrouplasso}
\end{align}}
In the same way as \eqref{eq:SVD_sparsegrouplasso} is equivalent to the use of an SVD as explained above, \eqref{eq:SVDs_sparsegrouplasso} is equivalent to working with SVDs of the column-scaled matrices $\{ \Pi^{\perp} X^{(j)} [S^{(j)}]^{-1} \}$.
\vskip1ex
\noindent As a baseline, we compute the solutions with only centering (i.e., \eqref{eq:scaling_sparsegrouplasso} without the $\{ S^{(j)} \}$).   
\vskip1ex
\noindent From an optimization perspective, formulation \eqref{eq:scaling_sparsegrouplasso} is most convenient as the form of the penalty with dependence only on $\beta$ gives rise to notable simplifications in block coordinate descent and proximal gradient methods (see, e.g., \cite{Simon2012b}). On the other hand, we consider \eqref{eq:SVDs_sparsegrouplasso} as most appropriate from the point of view of standardization as it combines column scaling for the $\ell_1$ penalty with a penalty on the fit per block rather than coefficients per block. Since it is not longer guaranteed that $\nnorm{\Pi^{\perp} X^{(j)} [S^{(j)}]^{-1} \wh{\beta}^{(j)}}_2 = \nnorm{\wh{\beta}^{(j)}}_2$,
\eqref{eq:scaling_sparsegrouplasso} is not equivalent to \eqref{eq:SVDs_sparsegrouplasso} in general. Figure \ref{fig:sparsegrouplasso_results} indicates that the difference is not substantial though. By contrast, the performance of \eqref{eq:SVD_sparsegrouplasso} is clearly inferior to both \eqref{eq:scaling_sparsegrouplasso} and \eqref{eq:SVDs_sparsegrouplasso}. The latter observation is expected since column re-scaling appears necessary to balance the $\ell_1$ penalty with respect to heterogeneity among level-wise frequencies.
\vskip1ex
\noindent Model fitting and evaluation are done in the same fashion as in Setting 1. Defining $\Lambda_{\textsf{w/o}}^{\text{lasso}} = 2^{\kappa} \sigma \sqrt{\log(p \cdot L)/n}$ and $\Lambda_{\textsf{w/}}^{\text{lasso}} = n^{-1/2}
\Lambda_{\textsf{w/o}}^{\text{lasso}}$ with $\kappa$ as in Setting 1, we consider
\begin{itemize}
\item $\lambda \in \Lambda_{\textsf{w/}}$ and $\lambda^{\text{lasso}} \in \Lambda_{\textsf{w/}}^{\text{lasso}}$ (for \underline{scaling} and \underline{SVD + scaling}),
\item $\lambda \in \Lambda_{\textsf{w/}}$ and $\lambda^{\text{lasso}} \in \Lambda_{\textsf{w/o}}^{\text{lasso}}$ (for \underline{SVD}).
\end{itemize}
Furthermore, we let $\tau \in \{0,0.1,\ldots,1\}$, and choose both $\kappa$ and $\tau$ such that the mean squared prediction error on a separate validation set is minimized.

\begin{figure}
\hspace*{4ex}\mbox{\includegraphics[width = 0.40\textwidth]{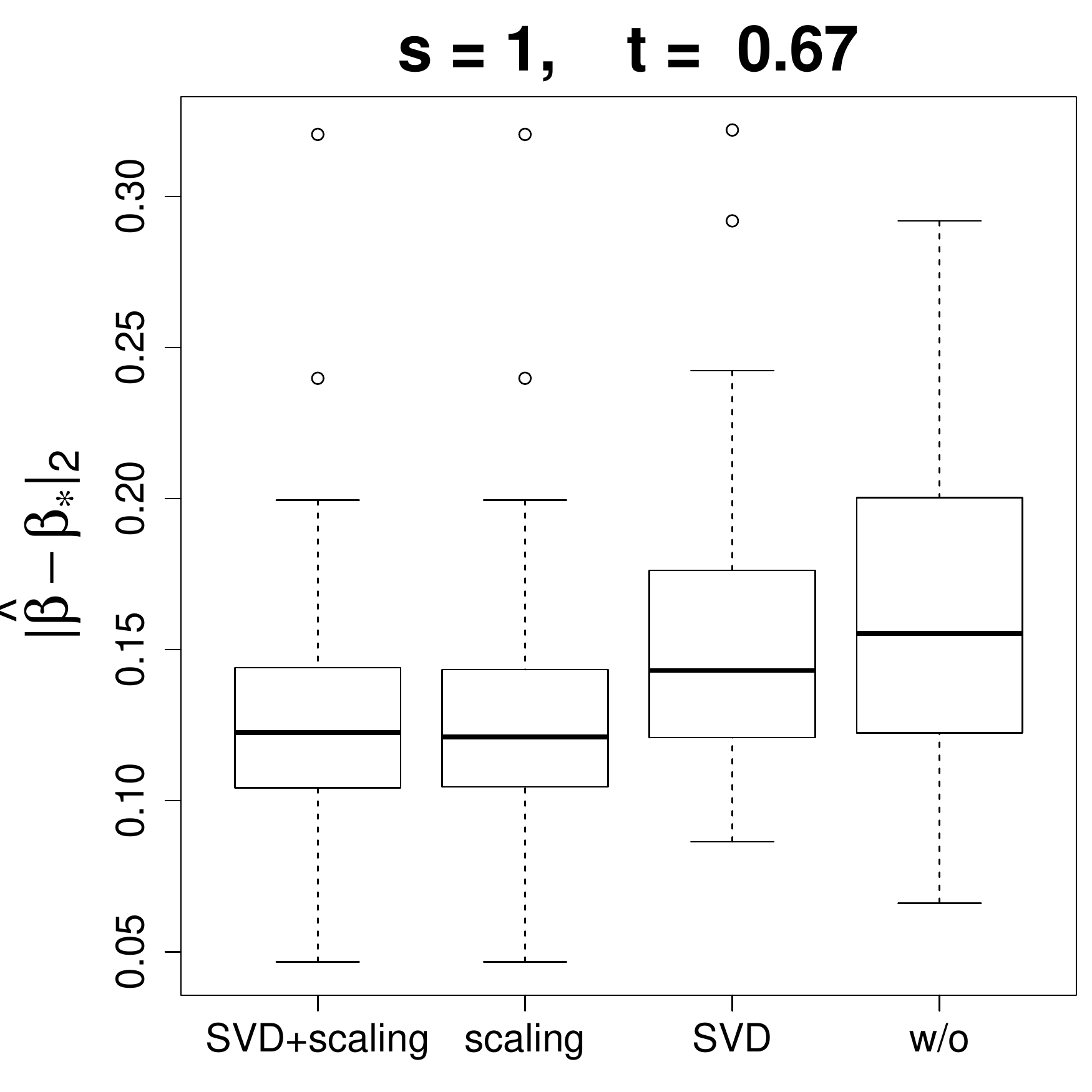} \hspace*{8ex} \includegraphics[width = 0.40\textwidth]{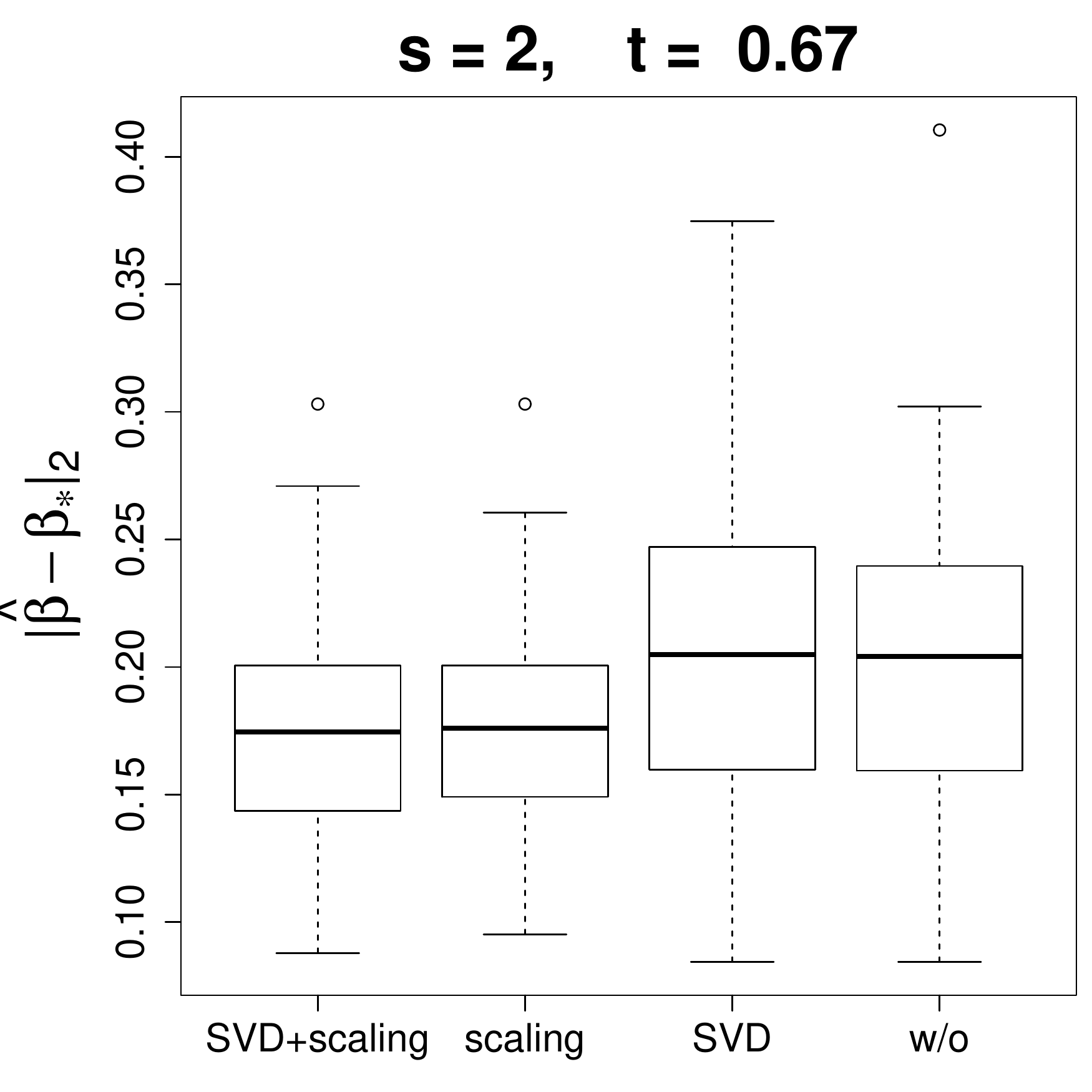} }

\caption{Boxplots of the estimation errors $\nnorm{\hat{\beta} - \beta_*}_2$ for Setting 2 (sparse group lasso) over $50$ independent replications. "w/o" refers to group lasso fits after centering only.}\label{fig:sparsegrouplasso_results}
\end{figure} 




\section{Conclusion}\label{sec:conclusions}
In this paper, we have re-visited block-wise standardization for the group lasso as discussed in Simon and Tibshirani \cite{Simon2012}. The setting
of the present paper concerns the important special case of categorical predictors. We have shown that the simplicity of the corresponding indicator matrices enables standardization without computationally expensive matrix decompositions. In the same vein, the case of first-order interactions can be reduced to complexity $O(L + M)$
as compared to $O(L \cdot M)$, where $L$ and $M$ denote the number of levels of two interacting categorical predictors. We have pointed out that our approach does not generalize to the sparse group lasso. Numerical studies presented herein confirm the positive effect of proper standardization on the estimation error.


\subsection*{Appendix: Proof of Theorem \ref{theo:interaction}}\label{app:theo1}
We start by noting that property \eqref{eq:standardization_condition_marginal} holds in view of the rationale underlying the treatment of the non-interaction case.

\noindent Turning to \eqref{eq:standardization_condition_interaction}, we observe that by construction
\begin{equation}\label{eq:nullspace_interaction_matrix}
S^{(1.2)} \text{range}(\mc{N}) = \text{null}(\texttt{P}^{\perp} X^{(1.2)} [S^{(1.2)}]^{-1}).
\end{equation}
Moreover, one computes that 
\begin{align}
  [X^{(1)}]^{\T} X^{(1.2)} [S^{(1.2)}]^{-1} = \begin{pmatrix}
          n_{11}^{1/2} & \ldots   & n_{1M}^{1/2}  & 0 & \ldots & \ldots & \ldots & \ldots & 0 \\
          0 & \ldots   &  0    &  n_{21}^{1/2} & \ldots   & n_{2M}^{1/2} & 0 & \ldots & 0 \\
           \vdots & \vdots   &  \ddots    &  &   &    & \ddots &  &  \\
                    0 & \ldots   &  \ldots    &  \ldots  & \ldots & 0   & n_{L1}^{1/2} & \ldots & n_{LM}^{1/2} \\
           \end{pmatrix},
\end{align}

  \begin{align}
    &[X^{(2)}]^{\T} X^{(1.2)} [S^{(1.2)}]^{-1} \notag \\
    &= \left( \begin{array}{ccccccccccccc}
                                                       n_{11}^{1/2} & 0  & \ldots  & 0 & n_{21}^{1/2} & 0 & \ldots &  0 & n_{L1}^{1/2} & 0 & \ldots & \ldots & 0  \\
                                                                    & n_{12}^{1/2} & 0 & \ldots & 0 & n_{22}^{1/2} & 0 & \ldots  &  0 & n_{L2}^{1/2} & 0 & \ldots & 0 \\
                                                                    &  & \ddots  &  &   & \ddots &  &  &   &  & \ddots &  &  \\
                                                                    &  &  n_{1M}^{1/2} & 0 & \ldots  & 0 & n_{2M}^{1/2} & 0 &  \ldots & 0 & 0 & \ldots & n_{LM}^{1/2}
           \end{array} \right)
\end{align}
Now consider $v \in \R^{L \cdot M}$ with 
\begin{equation*}
v = \left(v_{11} \ldots v_{1M} \; v_{21} \ldots v_{2M} \; \ldots \, \ldots v_{L1} \ldots v_{LM} \right)^{\T}. 
\end{equation*}
Then
{\small
\begin{equation}\label{eq:vsums}
  [X^{(1)}]^{\T} X^{(1.2)} [S^{(1.2)}]^{-1}  v = \begin{pmatrix}
    \sum_{m=1}^{M} n_{1m}^{1/2}  v_{1m}  \\
    \vdots \\
     \sum_{m=1}^{M} n_{Lm}^{1/2}   v_{Lm} \\
   \end{pmatrix}, \quad
   [X^{(2)}]^{\T} X^{(1.2)} [S^{(1.2)}]^{-1}  v = \begin{pmatrix}
    \sum_{l=1}^{L} n_{l1}^{1/2}  v_{l1}  \\
    \vdots \\
     \sum_{l=1}^{L} n_{lM}^{1/2}   v_{lM} \\
  \end{pmatrix}
\end{equation}}
With a reasoning similar to that in the proof of Lemma \ref{lem:Lim_Hastie}, one shows that
\begin{equation}\label{eq:nullspace_interaction_coefficient}
\nscp{\delta}{\wh{\beta}^{(1.2)}} = 0  \;\;\,\forall \delta \in \text{null}(\texttt{P}^{\perp} X^{(1.2)} [S^{(1.2)}]^{-1}). 
\end{equation}
Next, we observe that
\begin{align*}
  &S^{(1.2)} \mc{N}_1 =  \big\{\textstyle\sum_{m = 1}^M n_{1m}^{1/2}, \textstyle\sum_{m = M+1}^{2M} n_{2m}^{1/2}, \ldots, \textstyle\sum_{m = (L-1) \cdot M +1}^{L \cdot M} n_{Lm}^{1/2} \big\},  \\[1ex]
  &S^{(1.2)} \mc{N}_2 =  \big\{\textstyle\sum_{l = 1}^L n_{l1}^{1/2}, \textstyle\sum_{l = 1}^{L} n_{l2}^{1/2}, \ldots,  \textstyle\sum_{l = 1}^{L} n_{lM}^{1/2} \big\}. 
\end{align*}
Combining \eqref{eq:nullspace_interaction_matrix}, and \eqref{eq:nullspace_interaction_coefficient}, we conclude that
\begin{equation}\label{eq:orthogonality_interaction}
\nscp{S^{(1.2)} v}{\wh{\beta}^{(1.2)}} = 0  \;\; \forall v \in \mc{N}_1 \cup \mc{N}_2. 
\end{equation}
In particular, in light of \eqref{eq:vsums}
\begin{equation*}
[X^{(1)}]^{\T} X^{(1.2)} [S^{(1.2)}]^{-1} \wh{\beta}^{(1.2)} = 0, \qquad  [X^{(2)}]^{\T} X^{(1.2)} [S^{(1.2)}]^{-1} \wh{\beta}^{(1.2)} = 0. 
\end{equation*}
This implies that $\texttt{P}^{\perp} X^{(1.2)} [S^{(1.2)}]^{-1} \wh{\beta}^{(1.2)} = X^{(1.2)} [S^{(1.2)}]^{-1} \wh{\beta}^{(1.2)}$, and
by the fact that the columns of $X^{(1.2)} [S^{(1.2)}]^{-1}$ are orthonormal, we finally obtain \eqref{eq:standardization_condition_interaction}. It remains to check
\eqref{eq:nochange_in_fit} and \eqref{eq:sumtooneconstraints_satisfied}. The first property holds trivially. Property \eqref{eq:sumtooneconstraints_satisfied}
can be stated equivalently as $\nscp{v}{\wh{\theta}^{(1.2)}} = 0$ for all $v \in \mc{N}$ which is true by construction of $\wh{\theta}^{(1.2)}$.

\bibliographystyle{siam}
{
\bibliography{references_grouplassonote.bib}
}

\end{document}